\documentclass{amsart}
\usepackage{amssymb}
\usepackage{graphicx}
\usepackage{epstopdf}

\newtheorem{theorem}{Theorem}[section]
\newtheorem{lemma}[theorem]{Lemma}

\newtheorem{prop}[theorem]{Proposition}
\theoremstyle{definition}

\theoremstyle{remark}
\newtheorem{remark}[theorem]{Remark}
\newcommand{\dontprint}[1]{\relax}

\begin{document}
\title{Gaudin subalgebras and wonderful models}
\author[L. Aguirre]{Leonardo Aguirre}
\address{Department of mathematics,
ETH Zurich, 8092 Zurich, Switzerland}
\author[G. Felder]{Giovanni Felder}
\address{Department of mathematics,
ETH Zurich, 8092 Zurich, Switzerland}
\author[A. P. Veselov]{Alexander P. Veselov}
\address{ 	
Department of Mathematical Sciences\\ 
Loughborough University \\
Loughborough 
LE11 3TU \\
UK and Moscow State University, Russia}
\begin{abstract}
  Gaudin hamiltonians form families of $r$-dimensional abelian Lie
  subalgebras of the holonomy Lie algebra of the arrangement of
  reflection hyperplanes of a Coxeter group of rank $r$. We consider
  the set of principal Gaudin subalgebras, which is the closure in the
  appropriate Grassmannian of the set of spans of Gaudin hamiltonians.
  We show that principal Gaudin subalgebras form a smooth projective
  variety isomorphic to the De Concini--Procesi compactification of
  the projectivized complement of the arrangement of reflection
  hyperplanes.
\end{abstract}
\subjclass[2010]{Primary 81R12; Secondary 20F55, 14H70, 14N20}
\date{6 September 2014}
\keywords{}
\maketitle
\section{Introduction}\label{sec-1}
The hamiltonians of Gaudin's quantum integrable systems
are commuting elements of the tensor 
power $U\mathfrak g^{\otimes n}$ of the universal enveloping algebra of
a complex Lie algebra $\mathfrak g$ with a symmetric invariant
tensor $t\in \mathrm{Sym}^2(\mathfrak g)^{\mathfrak g}$. They depend
on $n$ distinct complex parameters and have the form:
\[
H_i=H_i(z_1,\dots,z_n)
=\sum_{1\leq j\leq n,j\neq i}\frac{t_{ij}}{z_i-z_j},\quad i=1,\dots,n.
\]
Here $t_{ij}$ is
the image of the invariant tensor
 by the embedding $U(\mathfrak
g)^{\otimes2}\to U(\mathfrak g)^{\otimes n}$ placing a tensor in the
$i$th and $j$th factor and filling with $1$. In the original context
studied by Gaudin $\mathfrak g=sl_2$ and thus the $H_i$ act on
$(\mathbb C^2)^{\otimes n}$ and can be interpreted as commuting
hamiltonians of an integrable quantum spin chain. Later the
Gaudin hamiltonians reappeared as the coefficients of the
Knizhnik--Zamolodchikov flat connections $d+\kappa^{-1}\sum H_idz$
on a trivial bundle on the complement of the diagonals in $\mathbb C^n$
with fiber a tensor product of $\mathfrak g$-modules.

The fact that the $H_i$ commute or, equivalently, the flatness of the
Knizhnik--Zamolodchikov connection, follows
from the Kohno--Drinfeld relations
\begin{equation}\label{eq-KD}
\begin{array}{cc}
{}[t_{ij},t_{kl}]=0,&\text{for  distinct $i,j,k,l$},
\\ 
{}[t_{ij},t_{ik}+t_{jk}]=0,&
\text{ for distinct $i,j,k$.}
\end{array}
\end{equation}
The Kohno--Drinfeld Lie algebra $\mathfrak t_n$ is the quotient of the
free Lie algebra in generators $t_{ij}=t_{ji}, 1\leq i<j\leq n$ by the
relations \eqref{eq-KD}. It is useful to consider the Gaudin
hamiltonians as commuting elements of $\mathfrak t_n$. Indeed the
Kohno--Drinfeld Lie algebra maps to other relevant Lie algebras,
giving rise to commuting elements there. In particular the map sending
$t_{ij}$ to the permutation $(ij)$ is a homomorphism from $\mathfrak
t_n$ to the group algebra of the symmetric group $S_n$, viewed as a
Lie algebra with the commutator bracket. As the relations are
homogeneous, the Kohno--Drinfeld Lie algebra comes with a grading
$\oplus_{i\geq 1} \mathfrak t_n^i$ so that $\mathfrak t_n^{1}$ is the
span of the generators. The linear span of the $H_i(z)$ form an
$n-1$-dimensional abelian subalgebra $G_n(z)$ of $\mathfrak t_n$ contained in
$\mathfrak t_n^1$, parametrized by the variety $\mathcal M_{0,n+1}$ of
$n$-tuples of distinct complex numbers up to affine
transformations. In our paper \cite{AFV} we studied the closure
$\mathcal G_n\subset \mathrm{Gr}(n-1,\mathfrak t_n^1)$ of this set of
abelian subalgebras in the Grassmannian of $n-1$-planes in $\mathfrak
t_n^1$. This closure contains in particular the span of the elements
$t_{12},t_{13}+t_{23},\dots, t_{1n}+\cdots+t_{n-1,n}$, mapping to the
Jucys--Murphy elements arising in the representation theory of the
symmetric group.

In \cite{AFV} we proved:
\begin{enumerate} 
\item[(i)] $\mathcal G_n$ is a nonsingular projective variety isomorphic to
  the Deligne--Mumford moduli space $\bar {\mathcal M}_{0,n+1}$ of stable
  rational curves with $n+1$ marked points, so that $G_n(z)$ corresponds
to $(\mathbb P^1,z_1,\dots,z_n,\infty)$.
\item[(ii)] $\mathcal G_n$ consists of all abelian subalgebras of $\mathfrak
  t_n$ of maximal dimension contained in the span $\mathfrak t_n^1$ of
  generators.
\end{enumerate}

The Kohno--Drinfeld Lie algebra is the special case of the $A_{n-1}$
system of a family of Lie algebras defined for general Coxeter systems
(and in fact for general central arrangements of hyperplanes). The aim
of this paper is to discuss the extension of our results to this more
general situation.  It turns out that (i) extends to general Coxeter
systems but (ii) does not hold in general.

Let $(H_\alpha)_{\alpha\in\Delta}$ be the arrangement of complexified
reflection hyperplanes $H_\alpha\in V$ of a finite Coxeter group,
defined by a set $\Delta\in V^*$ of linear forms. Assume that $\Delta$
span $V^*$ and let $r=\mathrm{dim}(V)$. Kohno's {\em holonomy Lie
  algebra} $\mathfrak t_\Delta$ \cite{Kohno} associated with the
arrangement is generated by $t_\alpha$, $\alpha\in\Delta$ with a
defining relation 
\begin{equation}\label{t_Delta}
[t_\alpha,\sum_{\beta\in W\cap\Delta} t_\beta]=0.
\end{equation}
for each pair $(\alpha,W)$ consisting of a linear form
$\alpha\in\Delta$ and a two-dimensional subspace $W\subset V^*$
containing $\alpha$.  In \cite{Kohno} this Lie algebra appears as the
Lie algebra of the unipotent completion of the fundamental group of
the complement of the union of the hyperplanes of the arrangement. Its
enveloping algebra is the quadratic dual of the Orlik--Solomon algebra
of the arrangement in the cases where the latter is quadratic, see
\cite{Yuzvinsky}. These relations insure the commutativity of the
Gaudin hamiltonians
\[
H(w)=\sum_{\alpha\in\Delta}\frac{\alpha(w)}{\alpha(z)}\, t_\alpha,\quad w\in V,
\]
for any fixed $z$ in the complement of the arrangement. We thus obtain
a family of abelian subalgebras $G_\Delta(z)=\mathrm{span}\{H(w),w\in
V\}$ of $\mathfrak t_\Delta$ contained in the span of the generators
and parametrized by the projectivization $\mathcal M_\Delta=\mathbb
P(V\smallsetminus\cup_{\alpha\in\Delta}H_\alpha)$ of the hyperplane
complement.
 
The relations \eqref{t_Delta} and the corresponding abelian subalgebras
$G_{\Delta}(z)$ arise in several different contexts:
\begin{enumerate}
\item Let $\Delta$ be a set of positive roots of a simple Lie algebra
  $\mathfrak g$ with root generators $e_\alpha,e_{-\alpha}$ normalized
  so that $\langle e_\alpha,e_{-\alpha}\rangle=1$ for an invariant
  bilinear form $\langle\ ,\ \rangle$.  Then $t_\alpha=e_\alpha
  e_{-\alpha}+e_{-\alpha} e_\alpha\in U(\mathfrak g)$ obey the
  relations \eqref{t_Delta} for any $g_\alpha\in U\mathfrak h$, the
  universal enveloping algebra of the Cartan subalgebra of $\mathfrak
  g$. The relations \eqref{t_Delta} (in the equivalent formulation of
  the commutativity of the generic subalgebras) are due to Vinberg,
  see \cite{Vinberg}, Theorem 2. The closure of the set of generic
  abelian subalgebras in the Grassmannian was described in the same
  paper \cite{Vinberg} as a set in terms of certain
  equivalence classes of paths in the Cartan subalgebra.  
\item The reflections $s_\alpha$ of a finite Coxeter group $G$ obey the
  relations \eqref{t_Delta}.  Indeed for any subspace $W\subset V^*$,
  the reflections $s_{\alpha}$ with $\alpha\in \Delta\cap W$ generate
  a subgroup which is also a Coxeter group and \eqref{t_Delta} is the
  statement that the sum of reflections of a Coxeter group is central
  in the group algebra, which is a consequence of the fact that the
  reflections form a conjugacy class. More generally we may take
  $t_\alpha=k_\alpha s_\alpha$ for a $G$-invariant function $\alpha\to
  k_\alpha$. The corresponding KZ connection $d+\sum k_\alpha
  s_\alpha d\log\alpha$ on the trivial bundle on the hyperplane
  complement with fibre $\mathbb C G$ was introduced by Cherednik
  \cite{Cherednik}.
\item The Kohno--Drinfeld Lie algebra is the special case where
  $\Delta$ is the root system $A_{n-1}$. Other cases of $\mathfrak t_\Delta$
  admitting a map to $U\mathfrak g^{\otimes n}$ for a semisimple Lie
  algebra $\mathfrak g$, include $\Delta=B_n$ for any $\mathfrak g$
  and $\Delta=D_n$ for $\mathfrak g=\mathfrak{sl}_N$, see
  \cite{Leibman}.
\end{enumerate}

Define a {\em Gaudin subalgebra} of $\mathfrak t_\Delta$ to be an
$r$-dimensional abelian Lie subalgebra contained in the span
$\mathfrak t^1_\Delta$ of generators. The spans $G(z)$ of Gaudin
hamiltonians are Gaudin subalgebras and so are their limits as
$z$ approaches the hyperplanes. We call these {\em principal
Gaudin subalgebras}, following Vinberg. 
Thus the set of principal Gaudin subalgebras
is the closure of the family $(G(z))_{z\in \mathcal M_\Delta}$ in the
Grassmannian of $r$-planes in $\mathfrak t^1_\Delta$.  
 
The main result of this paper is:

\begin{theorem}\label{t-1}
  Let $(H_\alpha)_{\alpha\in \Delta}$ be the arrangement of reflection
  hyperplanes of a Coxeter system of rank $r$.  The set of principal
  Gaudin subalgebras is a smooth subvariety of the Grassmannian
  $\mathrm{Gr}(r,\mathfrak t^1_\Delta)$ and is isomorphic to the De
  Concini--Procesi compactification $\bar {\mathcal M}_\Delta$
  \cite{dCP} of $\mathbb
  P(V\smallsetminus\cup_{\alpha\in\Delta}H_\alpha)$.
\end{theorem}

The paper is organized as follows: in Section \ref{sec-2} we recall
the definition and description of the De Concini--Procesi
compactification $\bar{\mathcal M}_\Delta$ of the complement of an
arrangement of hyperplanes. In Section \ref{sec-3} we give an
embedding of $\bar M_\Delta$ in the Grassmannian
$\mathrm{Gr}(r,\mathfrak t^1_\Delta)$. These constructions apply to
any central arrangement. We then specialize to the case of Coxeter
arrangements and formulate Theorem \ref{t-2} in Section \ref{sec-4},
which is a more precise version of our result and implies Theorem
\ref{t-1}.  

We show that our statement does not hold for general arrangements
by giving a simple counterexample, taken from
\cite{Leo}, where the case of general arrangement is considered.

A real version of the theorem also holds and leads to the moduli space tessellated by convex polyhedra known as De Concini-Procesi associahedra \cite{dCP}, which are generalisations of Stasheff polytopes determined by the corresponding Coxeter graphs \cite{CD, TL}. 

We conclude with the discussion of non-principal Gaudin subalgebras, using $B_n$ case as an example.

\section{De Concini--Procesi compactification}\label{sec-2}
Let $(H_\alpha)_{\alpha\in\Delta}$ be a central arrangement of
hyperplanes of rank $r$, i.e., a finite family of hyperplanes through
the origin in an $r$-dimensional complex vector space $V$.  It will be
convenient to label the family by a set $\Delta\subset V^*$ of
non-zero linear forms vanishing on the hyperplanes and work with the
set of vectors $\Delta$ rather than with the hyperplanes.  We 
assume for convenience that $\Delta$ is irreducible, in the sense that
$V$ cannot be written as a non-trivial direct sum of two subspaces
whose union contains $\Delta$.  The projectivized complement
${\mathcal M}_\Delta=\mathbb
P(V\smallsetminus\cup_{\alpha\in\Delta}H_\alpha)$ admits a smooth
compactification $\bar {\mathcal M}_\Delta$, the ``wonderful model''
of De Concini and Procesi.

\subsection{The partially ordered set of flats}
We denote by $\langle A\rangle$ the linear span of a subset $A$ of $V^*$ and
call a non-empty subset $A\subset \Delta$ a {\em flat} if it contains all
elements of $\Delta$ in its span, i.e., if $\langle A\rangle\cap
\Delta=A$.
A flat $A$ is called reducible if there are flats $A_1,A_2$ such
that $A_1\cup A_2=A$ and $\langle A\rangle= \langle
A_1\rangle\oplus\langle A_2\rangle$, otherwise irreducible.
Irreducible flats form a partially ordered set by inclusion with
maximal element $\Delta$ and minimal elements the one-point sets
$\{\alpha\}$, $\alpha\in \Delta$. 

\subsection{The De Concini--Procesi compactification}
If $A$ is a subset of $V^*$ we denote by $A^\perp=\cap_{\alpha\in
  A}H_\alpha$ the orthogonal complement. The natural projection $V\to
V/A^\perp$ restricts to a map $V\smallsetminus A^\perp\to
V/A^\perp\smallsetminus\{0\}$ and induces a projection
$V\smallsetminus A^\perp\to\mathbb P( V/A^\perp)$.  Thus we have a map
\[
j_\Delta\colon
{\mathcal M}_\Delta\hookrightarrow \prod_{A}\mathbb P(V/A^\perp)
\]
with product taken over all irreducible flats. The map is injective since
$A=\Delta$ is irreducible and thus $\mathbb P(V/\Delta^\perp)=\mathbb
P(V)$ belongs to the product.

The De Concini--Procesi compactification $\bar {\mathcal M}_\Delta$ is
by definition the closure of the image of $j_\Delta$. It is a smooth
projective variety containing $j_\Delta({\mathcal M}_\Delta)\cong
{\mathcal M}_\Delta$ as a Zariski open subset whose complement is the
union of smooth divisors with normal crossings \cite{dCP}. 
\subsection{Nested sets}
A {\em nested set} for $\Delta$ is a set $S$ of irreducible flats of
$\Delta$ so that for any subset $\{A_1,\dots,A_k\}$ of $S$ consisting
of pairwise non-comparable flats, we have $ \langle
\cup_{i=1}^kA_i\rangle=\langle A_1\rangle\oplus\cdots\oplus\langle
A_k\rangle$. Nested sets are partially ordered by inclusion.
\begin{lemma}\label{l-AS} (\cite{dCP}, p.~500) Let $S$ be a nested set
  containing $\Delta$ and $\alpha\in \Delta$. Then the subsets of $A$
  containing $\alpha$ are linearly ordered.  Thus there is a
  unique minimal flat $A_S(\alpha)\in S$ such that $\alpha\in A$.
\end{lemma}
We are particularly interested in {\em maximal nested sets}. They have
the property (\cite{dCP}, Proposition 1.1) that for every $A\in S$ the
sets $A_1,\dots, A_k\in S$ that are maximal proper subsets of $A$ obey
$\sum\mathrm{dim}(A_i)=\mathrm{dim}(A)-1$. Thus one can pick a vector
$\alpha_A\in A\smallsetminus \cup_iA_i$ and obtain a basis of $V$
labeled by $S$. Such bases are called {\em adapted bases} for the
maximal nested set $S$. They have the property that their intersection
with any $A\in S$ is a basis of $\langle A\rangle$.
\subsection{An open cover of $\bar {\mathcal M}_\Delta$}
The De Concini--Procesi variety $\bar {\mathcal M}_\Delta$ admits a
cover by open affine subsets $U_S$ labeled by maximal nested sets $S$.
For every maximal nested set $S$ the open subset $U_S$ is defined as
the set of $z\in \bar {\mathcal M}_\Delta$ such that for every $A\in
S$ the projection $z_A\in \mathbb P(V/A^\perp)$ obeys
\[
\forall \alpha\in A\smallsetminus \cup_{i=1}^kA_i:\alpha(z_A)\neq 0, 
\]
where $A_1,\dots,A_k\in S$ are the maximal elements properly contained
in $A$.

\begin{theorem}\label{t-dCP}(De Concini--Procesi, \cite{dCP} Sec.~1.1)
\
\begin{enumerate}
\item[(i)] The sets $U_S$, where $S$ runs over the maximal nested sets
  for $\Delta$, form an open covering of $\bar {\mathcal M}_\Delta$.
\item[(ii)] The natural projection $U_S\to \prod_{A\in S}\mathbb
  P(V/A^\perp)$ is a closed embedding.
\end{enumerate}
\end{theorem}

\section{A map to the Grassmannian}\label{sec-3}
Let $\Delta$ be a set of non-collinear vectors in an $r$-dimensional
complex vector space $V$. Kohno's holonomy Lie algebra $\mathfrak t_\Delta$ 
is the graded Lie algebra with generators $t_\alpha$ of degree 1 labeled
by $\alpha\in\Delta$ and relations (\ref{t_Delta}).
Let $\mathfrak t^1_\Delta$ be span of generators.
For our present purpose it is just a vector space with basis
$(t_\alpha)_{\alpha\in\Delta}$. We construct a map from $\bar
{\mathcal M}_{\Delta}$ to the Grassmannian $\mathrm{Gr}(r,\mathfrak
t^1_\Delta)$ of $r$-dimensional subspaces of $\mathfrak t^1_\Delta$
restricting to a locally closed embedding of ${\mathcal M}_\Delta$.
The uncompactified ${\mathcal M}_\Delta$ is embedded by sending a
point in the hyperplane complement to the span of the corresponding
Gaudin hamiltonians:
\begin{align*}
  i\colon {\mathcal M}_\Delta&= \mathbb P(V\smallsetminus
  \cup_{\alpha\in\Delta} H_\alpha)
  \to \mathrm{Gr}(r,\mathfrak t^1_\Delta)\\
  z&\mapsto \left\{
    \sum_{\alpha\in\Delta}\frac{\alpha(w)}{\alpha(z)}t_\alpha, w\in V
  \right\}.
\end{align*}
\begin{prop} \label{p-0} The embedding $i\colon {\mathcal M}_\Delta\to
  \mathrm{Gr}(r,\mathfrak t^1_\Delta)$ extends uniquely to a map
 $\bar i\colon
  \bar {\mathcal M}_\Delta\to \mathrm{Gr}(r,\mathfrak t^1_\Delta)$.
\end{prop}
The uniqueness is obvious since ${\mathcal M}_\Delta$ is a Zariski open
subset.  We prove the existence by analyzing the map $i$ on each of
the open sets $U_S$ labeled by the maximal nested sets $S$. We also
prove some further properties of $i|_{U_S}$ that will be useful to
prove that $\bar i$ is an embedding in the case of Coxeter systems.

Let $S$ be a nested set for $\Delta$ and let $B=(\alpha_A)_{A\in S}$
be an adapted basis for $S$.  Let $V_B$ be the open subset of
$\mathrm{Gr}(r,\mathfrak t^1_\Delta)$ consisting of $r$-planes whose
projection onto the span of $t_\beta$ with $\beta\in B$ is
surjective. Such $r$-planes are given by systems of equations of the
form
\begin{equation}\label{e-coordinates}
  t_\alpha^*-\sum_{\beta\in B} c_{\alpha,\beta}t_\beta^*=0,
  \quad \alpha\in\Delta\smallsetminus B.
\end{equation}
where $t_\alpha^*$ is the basis of $(\mathfrak t^1_\Delta)^*$ dual to
$t_\alpha$ and $c_{\alpha,\beta}$ are arbitrary scalar coefficients.
Let $\mathrm{supp}_B(\alpha)\subset B$ be the set of basis elements
occurring with nonzero coefficient in the expression of
$\alpha\in\Delta$ as linear combination of $B$ and
let $V_B^0$ be the subspace of the affine space $V_B$ given by
the equations
\begin{equation}\label{V0}
  c_{\alpha,\beta}=0,\qquad \text{if $\beta\not\in \mathrm{supp}_B(\alpha)$}.
\end{equation}
Let $\alpha=\sum_{\beta\in B}n_{\alpha,\beta}\beta$ be the expression
of $\alpha$ as linear combination of the basis. Then for
$z=(z_A)_{A\in\Delta}$ we set
\[
  i_S(z)=\cap_{\alpha\in \Delta\smallsetminus
  B}\mathrm{Ker}(\alpha(z_{A_S(\alpha)})t_\alpha^*-\sum_{\beta\in
  B}n_{\alpha,\beta}\beta(z_{A_S(\alpha)})t_\beta^*)\subset V.
\]
Here $A_S(\alpha)$ denotes the smallest $A\in S$ containing $\alpha$
(cf.~Lemma \ref{l-AS}).

\begin{lemma} Let $S$ be a maximal nested set for $\Delta$.
\begin{enumerate}
\item[(i)] For all $z\in U_S$, $i_S(z)\in V^0_B$.
\item[(ii)] The map $i_S$ coincides with $i$ on $U_S\cap {\mathcal M}_{\Delta}$.
\end{enumerate}
\end{lemma} 
\begin{proof}
  (i) To show that $i_S(U_S)\subset V_B$ we need to show that
  $\alpha(z_{A_S(\alpha)})\neq 0$ for $z\in U_S$. By construction,
  $\alpha\in A_S(\alpha)$ and $\alpha$ is not contained in any other
  $A\in S$ contained in $A_S(\alpha)$. Thus, by definition of $U_S$, $
  \alpha(z_{A_S(\alpha)})\neq 0$. It is clear that the coefficients
  $c_{\alpha,\beta}$ in the image of $i_S$ obey \eqref{V0}.

  (ii) If $(z_A)_{A\in\Delta}\in {\mathcal M}_\Delta$ then $z_A$ is the image of
  $z\in V\smallsetminus\cup_\alpha H_\alpha$ and $i_S(z)$ is an
  $r$-dimensional subspace containing the Gaudin hamiltonians
  $H(w)=\sum_{\gamma\in \Delta}\frac{\gamma(w)}{\gamma(z)}t_\gamma$,
  $w\in V$. Indeed
  \[
  (\alpha(z)t_\alpha^*-\sum_{\beta\in B}n_{\alpha,\beta}\beta(z)
  t_\beta^*)(H(w))=\alpha(w)-\sum_{\beta}n_{\alpha,\beta}\beta(w)=0.
  \]
  Since the Gaudin hamiltonians form an $r$-dimensional vector space,
  their span coincides with the image of $i_S$.
\end{proof}
The maps $i_S$ are thus the restrictions of a map $\bar i\colon \bar
{\mathcal M}_\Delta\to \mathrm{Gr}(r,\mathfrak t^1_\Delta)$ coinciding with $i$
on ${\mathcal M}_\Delta$, proving Proposition \ref{p-0}.
\begin{lemma}\label{l-2} 
  The $r$-planes in the image of $\bar i\colon \bar {\mathcal M}_\Delta\to
  \mathrm{Gr}(r,\mathfrak t^1_\Delta)$ all contain the vector
  $C_\Delta=\sum_{\alpha\in\Delta} t_\alpha$.
\end{lemma}
\begin{proof}
 \[
 (\alpha(z)t_\alpha^*-\sum_{\beta\in
   B}n_{\alpha,\beta}\beta(z)t_\beta^*)(C_\Delta)=\alpha(z)
 -\sum_{\beta}n_{\alpha,\beta}\beta(z)=0.
 \]
\end{proof}
\begin{remark} 
  The vector $C_\Delta$ spans the center of the holonomy algebra
  $\mathfrak t_\Delta$.
\end{remark}

\section{The case of Coxeter systems}\label{sec-4}
Let us consider the special case of the arrangement of reflection
hyperplanes of an irreducible Coxeter group $G$. In this case we can
describe maximally nested sets combinatorially in terms of subsets of
the set of nodes of the Coxeter diagrams. We start by recalling
some facts about Coxeter systems and explain the properties we
will need. 

\subsection{Coxeter root systems}
We follow \cite{Humphreys}.  Let $E$ be an $r$-dimensional euclidean vector
space and for $\alpha\in E\smallsetminus\{0\}$ denote by $s_\alpha$
the orthogonal reflection with respect to the hyperplane with normal
vector $\alpha$. A (Coxeter) root system is a finite set $\Phi$ of
nonzero vectors in $E$, called roots, such that
\begin{enumerate}
\item For all $\alpha\in\Phi$, $\mathbb R\alpha\cap \Phi=\{\alpha,-\alpha\}$
\item For all $\alpha\in\Phi$, $s_\alpha\Phi=\Phi$.
\end{enumerate}
A root system is called reducible if $E$ is the orthogonal direct sum
of two subspaces of positive dimension whose union contains $\Phi$, otherwise
irreducible.  
The group $G$ generated by the reflections $s_\alpha$, $\alpha\in
\Phi$ is then a finite Coxeter group and any finite Coxeter group is
of this form. Any linear form that does not vanish on roots
decomposes $\Phi$ into positive and negative roots, according to the
sign it takes on roots.  There is then a unique basis of $E$
consisting of roots with the property that all positive roots are linear
combinations of basis vectors
with nonnegative coefficients.  Such a basis is called
basis of simple roots and the corresponding reflections simple
reflections. The Coxeter group acts simply transitively on bases of
simple roots. Angles between simple roots are of the form
$\pi(1-1/m)$ with $m\in\mathbb Z_{\geq 2}$. These data are encoded in the
Coxeter graph, a graph with labeled edges with
simple roots as vertices and an edge with label $m$ connecting
non-orthogonal simple roots making an angle $\pi(1-1/m)$, $(m\geq3)$.
Irreducible root systems have connected Coxeter graphs.
\begin{lemma}\label{l-hr}
  Let $\Phi\subset E$ be an irreducible Coxeter systems and $B$ a
  basis of simple roots. Then there is root $\theta\in \Phi$ which is
  a linear combination of simple roots with positive coefficients.
\end{lemma}
\begin{proof}
  Let us order the simple roots $B=\{\alpha_1,\dots,\alpha_r\}$ so
  that the subgraphs with vertices $\alpha_1,\dots,\alpha_i$ are
  connected for all $i=1,\dots, r$. This is clearly possible since the
  Coxeter graph is connected. This condition means that for each $i>1$
  there is a $j<i$ such that the inner product $(\alpha_i,\alpha_j)$
  is non-zero and in fact (since all angles are obtuse) negative. Let
  $s_i$ denote the corresponding simple reflections. We claim that,
  for all $i$, $v_i:=s_i\cdots s_3s_2\alpha_1$ is a linear combination
  of $\alpha_1,\dots,\alpha_i$ with positive coefficients.  This is
  obvious for $i=1$. Assuming by induction that $v_i$ has the required
  property, we see that in
  \[
  v_{i+1}=s_{i+1}v_i=v_i-2\frac{(\alpha_{i+1},v_i)}
  {(\alpha_{i+1},\alpha_{i+1})}\alpha_{i+1}
  \]
  the inner product $(\alpha_{i+1},v_i)$ is negative, proving the
  induction step. Thus we may take $\theta=s_r\cdots s_2\alpha_1$.
\end{proof}

\subsection{De Concini--Procesi compactification
for Coxeter arrangements}
Let $\Phi\subset E$ be a root system and $G$ the corresponding
group generated by reflections.
Fix a basis $B$ of simple roots and let $\Delta=\Phi_+$ be the corresponding
set of positive roots. We view $\Delta$ as a subset of the dual of
the complexification of $E\simeq E^*$, so that the complexified
reflection hyperplanes $H_\alpha\subset V$ are the kernels of
$\alpha\in\Delta$. Then $G$ acts on the set of hyperplanes and thus
on the set of nested sets. We say that two nested sets are $G$-equivalent
if they are related by an element of $G$.

\begin{prop}\label{p-1}
  Let $\Phi\subset V^*$ be an irreducible root system with
  Coxeter group $G$. Fix a basis $B$ of simple roots.  Then every
  nested set $S$ is $G$-equivalent to one where each subset in $S$ is
  spanned by simple roots. Thus, up to $G$-equivalence, every maximal
  nested set has an adapted basis consisting of simple roots.
\end{prop}

This proposition is proved in \cite{dCP}, Sec.~3.1, in the case of
Weyl groups. The same proof applies to the Coxeter case.

\begin{remark} A consequence of Proposition
  \ref{p-1} is that, up to $G$-equivalence, maximal nested sets for
  Coxeter arrangement are obtained from maximal nested sets on Coxeter
  graphs, whose vertices are naturally labeled by simple roots. 
  By definition, a nested set on a graph $\Gamma$ is a set of nonempty
connected subgraphs of $\Gamma$ which are pairwise
either disjoint or contained one into the other. 
Here two subgraphs are called disjoint if there are no edges joining vertices from these two subgraphs.
To a connected
subgraph $I$ of $\Gamma$ one associates the flat generated by simple roots labelling the vertices of $I$. 
\end{remark}

\begin{theorem}\label{t-2}
  Let $(H_\alpha)_{\alpha\in\Delta}$ be the set of reflection
  hyperplanes of a Coxeter system of rank $r$. Then the
  map $\bar i\colon \bar {\mathcal M}_\Delta\to\mathrm{Gr}(r,\mathfrak
  t^1_\Delta)$ is a closed embedding.
\end{theorem}
\begin{proof}
  It is sufficient to consider the case of irreducible Coxeter root
  systems.  By Lemma \ref{l-2} the image of $\bar i$ is contained in
  the subvariety $G'$ of $r$-planes containing $C_\Delta$.  It is
  sufficient to show that for each nested set $S$ we can choose an
  adapted basis $B$ so that $\bar i\colon U_S\to V_B^0\cap G'$ has a
  left inverse $\varphi_S\colon V_B^0\cap G'\to U_S$, i.e., a regular
  map such that $\varphi_S\circ i=\mathrm{id}_{U_S}$. The subset
  $V_B^0$ is the subspace of the affine space $V_B$ consisting of
  $r$-dimensional subspaces determined by equations of the form
  \[
  t_\alpha^*-\sum_{\beta\in\mathrm{supp}_B(\alpha)}c_{\alpha,\beta}t_\beta^*=0,
  \quad \alpha\in \Delta\smallsetminus B.
  \]
  It has coordinates $c_{\alpha,\beta}$,
  $\beta\in\mathrm{supp}_B(\alpha)$ and $V_B^0\cap G'$ is the affine
  subspace of $V^0_B$ given by the equations
  \begin{equation}\label{e-101}
    \sum_{\beta\in\mathrm{supp}_{A_S}(\alpha)} c_{\alpha,\beta}=1,\quad 
    \alpha\in \Delta\smallsetminus B.
  \end{equation}
  Let us view $U_S$ as a subset of $\prod_{A\in S}
  \mathbb{P}(V/A^\perp)$ via Theorem \ref{t-dCP} (ii).  The
  coordinates of $i(z)$ are given in terms of the projections
  $z_A\in\mathbb P(V/A^\perp)$ of $z\in U_S$ by
  \begin{equation}\label{e-102}
  c_{\alpha,\beta}=n_{\alpha,\beta}\frac{\beta(z_A)}{\alpha(z_A)},\quad
  \alpha\in\Delta\smallsetminus B, \quad\beta\in B.
  \end{equation}
  Here $A=A_S(\alpha)$ is the smallest $A\in S$ containing $\alpha$,
  so that, by definition of $U_S$, the denominator does not vanish.
  Pick for every $A\in S$ a root $\theta=\theta_A\in A$ with
  $n_{\theta,\beta}\neq 0$ for all $\beta\in B\cap A$. Such a root
  exists by Lemma \ref{l-hr}.  In particular $A$ is the smallest set
  in $S$ containing $\theta_A$.  For any $W\in V_B\cap G'$ with
  coordinates $(c_{\alpha,\beta})$ define the component $z_A$ of
  $z=\varphi(W)$ as the line through $\tilde z_A\in V/A^\perp$ such
  that
  \[
  \beta(\tilde z_A)=\frac{c_{\theta,\beta}}{n_{\theta,\beta}},\quad
  \theta=\theta_A,\quad \beta\in B\cap A
  \]
  As $B\cap A$ is a basis of $(V/A^\perp)^*$ this uniquely defines a
  vector $\tilde z_A\in V/A^\perp$. By \eqref{e-101} it is normalized
  by $\theta(\tilde z_A)=1$.  In particular $\tilde z_A$ is nonzero
  and defines a line $z_A\in \mathbb P(V/A^\perp)$ obeying
  \eqref{e-102} for $\alpha=\theta$. Repeating this for all $A\in S$
  we obtain a map
  \[
  \varphi_S \colon V_B\to \prod_{A\in S}\mathbb P(V/A^\perp).
  \]
  By construction $\varphi_S\circ i=\mathrm{id}_{U_S}$.
\end{proof} 

\begin{remark}
  One may ask whether the map $\bar i$ is a closed embedding for
  general arrangement $\Delta$. This is not the case. The following
  counterexample is taken from \cite{Leo}, where a necessary and
  sufficient condition for having a closed embedding is formulated.
  Let $V^*$ have a basis $\{\alpha_1,\alpha_2,\alpha_3\}$ and let
  $\Delta=\{\alpha_1,\alpha_2,\alpha_3,
  \alpha_1+\alpha_2,\alpha_1+\alpha_3\}$. The projectivized
  arrangement of hyperplanes is then given by five lines in the
  projective plane forming the sides and a diagonal of a
  quadrilateral.  There are exactly two points $P,Q$ at which three
  lines meet. The De Concini--Procesi compactification $\bar {\mathcal
    M}_\Delta$ is the blowup $\widehat{\mathbb P^2}$ of $\mathbb P^2$
  at these two points. It surjects under $\bar i$ to the closure
  $\overline{i({\mathcal M}_\Delta)}$ in the Grassmannian of the span
  of Gaudin hamiltonians, which is isomorphic to $\mathbb P^1\times
  \mathbb P^1$: $(\lambda_1{:}\lambda_2,\mu_1{:}\mu_2)\in \mathbb
  P^1\times \mathbb P^1$ corresponds to the Gaudin subalgebra
  $\mathrm{span}(C_\Delta,\lambda_1t_{\alpha_2}+
  \lambda_2t_{\alpha_1+\alpha_2},\mu_1t_{\alpha_3}+\mu_2t_{\alpha_1+\alpha_3})$.
  We have the diagram of birational morphisms
  \[
  \mathbb P^2\leftarrow \bar {\mathcal M}_\Delta \cong \widehat{\mathbb P^2}
  \stackrel{\bar i}\rightarrow\overline{i({\mathcal M}_\Delta)}\cong \mathbb
  P^1\times\mathbb P^1
  \]
  The map $\bar i$ is the blowdown of the proper transform of the line
  through $P$ and $Q$. It maps a curve to a point and is therefore not
  injective.  \dontprint { (*** Internal comment:
    What goes wrong in the proof in this case is that an adapted basis
    of a nested set containing $\{\alpha_1\}$ must contain
    $\alpha_1$. For such a basis $B$ there is no vector in $\Delta$
    with support $B$. ***) }
\end{remark}

\section{Real version and graph-associahedra}

All the considerations above work over reals as well. 
The corresponding real wonderful models $\bar M_\Delta(\mathbb R)$ were studied in the $A_{n-1}$ case by Kapranov \cite{Kap} and in the general Coxeter case by De Concini and Procesi \cite{dCP} (see also Gaiffi \cite{Gaiffi}). 

They can be described as an iterated real blow-up of the projective space and can be glued from $|G|/2$ copies of the following generalisations of Stasheff polytope, or associahedron.
The corresponding convex polytopes are known as {\it graph-associahedra} \cite{CD}, or {\it De Concini--Procesi associahedra} \cite{TL}.
They have the following explicit description \cite{Devadoss1, TL} inspired by \cite{Ma,St}. 

Let $\Gamma$ be a connected graph with set of vertices $B$ and consider the set $S_{\Gamma}$ of all connected subgraphs $I \subset \Gamma$ excluding $\Gamma$.
The corresponding graph-associahedra can be defined as the following convex polytope in $\mathbb R^{|B|}$ with coordinates labeled by the vertices  of $\Gamma$
\[
P_{\Gamma}=\{ x \in \mathbb R^{|B|}: \,\, \sum_{\alpha \in B} x_{\alpha}=3^{|B|}, \,\,  \sum_{\alpha \in I} x_{\alpha}\geq 3^{|I|},  \,\, I \in S_{\Gamma}\}.
\]
It is known \cite{CD, TL} that $P_{\Gamma}$ is a simple, convex polytope, whose face poset is isomorphic to the poset of the nested sets on $\Gamma$ (see the definition in remark 4.3).
In particular, the vertices of $\Gamma$ corresponds to the maximal nested sets. 
One can also describe it as the iterated truncation of a simplex \cite{CD}.

In the case when $\Gamma$ is a path with $r$ vertices this gives a particular realisation of the Stasheff polytope $K_{r+1}$, 
which was initially introduced only combinatorially \cite{Stas}.

\begin{figure}[h]
\centerline{ \includegraphics[width=4.5cm]{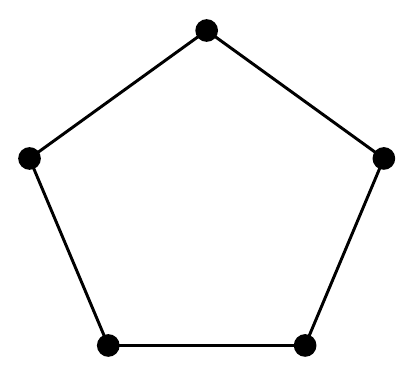}  \hspace{5pt}  \includegraphics[width=4.5cm]{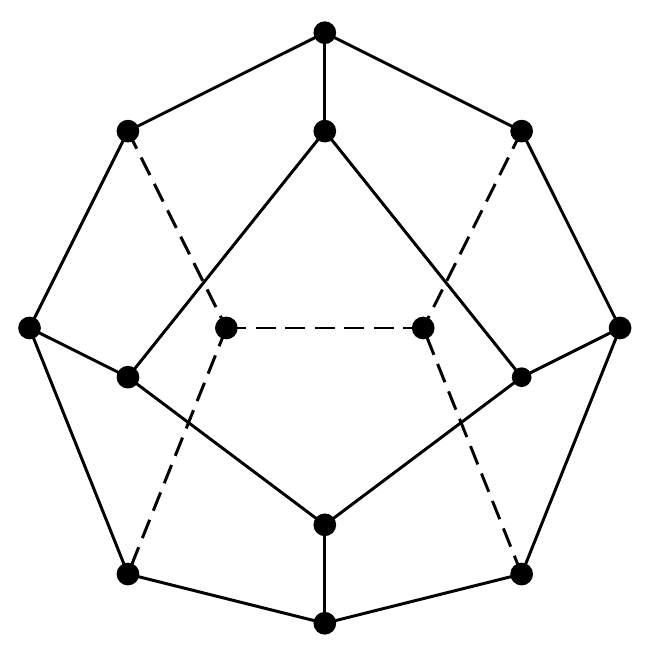}}
\caption{Stasheff polytopes $K_4$ and $K_5$.} \label{etas}
\end{figure}

When $\Gamma=\Gamma_G$ is a Coxeter graph we have the De Concini--Procesi associahedron $P_G$.

Note that the labelling of the angles on the Coxeter graph does not affect the definition of the polytope, so the polytopes $P_G$
for the arrangements of types $A_n$ and $B_n=C_n$ are the same Stasheff polytopes $K_{n+1}$. 

The space $\bar M_\Delta(\mathbb R)$ is tessellated by $|G|/2$ copies of $P_G$, 
which is isomorphic to the closure of the real points corresponding to one Coxeter chamber \cite{dCP}.
The details of the gluing can be found in \cite{CD}. In particular, in the $A_3$ and $B_3$ cases we have two different non-oriented surfaces glued from 12 and 24 pentagons, which topologically are the connected sums of 5 and 8 copies of the real projective plane, respectively.

As it follows from our results modulo this gluing one can use $P_G$ as a geometric chart of all principal Gaudin subalgebras in the real case (see the discussion of the $A_n$ case in \cite{SV}).

\section{Non-principal Gaudin subalgebras}

It is natural to ask if there are Gaudin subalgebras outside the principal family.
We have shown in \cite{AFV} that in the $A_n$ case the answer is negative.
However, already in the $B_n$ case not all Gaudin subalgebras are principal.
 
Let 
$$r_i, t_{ij}=t_{ji}, s_{ij}=s_{ji}, 1\leq i \le j \leq n$$ be the generators of the $B_n$ holonomy Lie algebra, corresponding to the hyperplanes $z_i=0, \, z_i-z_j=0, \, z_i+z_j=0$, respectively.
The defining relations have the form for pairwise distinct $i,j,k,l$ 
$$[r_i, r_j+t_{ij}+s_{ij}]=0,$$
$$[t_{ij}, r_i+r_j+s_{ij}]=0,$$
$$[s_{ij}, r_i+r_j+t_{ij}]=0$$
for $B_2$ subsystems,
$$[t_{ij}, t_{ik}+t_{jk}]=0,$$
$$[t_{ij}, s_{ik}+s_{jk}]=0,$$
$$[s_{ij}, t_{ik}+s_{jk}]=0$$
for $A_2$ subsystems, and
$$[r_i, t_{jk}]=[r_i, s_{jk}]=0,$$
$$[t_{ij}, t_{kl}]=[t_{ij}, s_{kl}]=[s_{ij},s_{kl}]=0$$
for $A_1\times A_1$ subsystems.

The principal family of Gaudin subalgebras is the closure of the family
\begin{equation}
\label{B}
G(z_1,\dots, z_n)=\{\sum_{i=1}^n\frac{a_i}{z_i}r_i+\sum_{i<j}^n\frac{a_i-a_j}{z_i-z_j}t_{ij}
+\sum_{i<j}^n\frac{a_i+a_j}{z_i+z_j}s_{ij},\, a \in \mathbb C^n\},
\end{equation}
but it is easy to see that there is another family of Gaudin subalgebras

\begin{equation}
\label{A}
G^{A}(z_1,\dots, z_n)=\{\sum_{i=1}^n\frac{a_i}{z_i}r_i+\sum_{i<j}^n\frac{a_i-a_j}{z_i-z_j}(t_{ij}+s_{ij}), \, a \in \mathbb C^n\}.
\end{equation}
The closure of the second family is isomorphic to the $A_{n}$-type wonderful model, which is the Deligne--Mumford--Knudsen moduli space 
$\bar M_{0,n+2}.$ 

So the variety of $B_n$-Gaudin subalgebras contains at least two $(n-1)$-dimensional subvarieties. For $n\geq 3$ they are irreducible components but they are not the only ones, as explained in \cite{Leo}. We close by examining the cases $n=2$ and $n=3$. Details can be found in \cite{Leoth}.

\emph{The $B_2$-Case: } The variety of Gaudin subalgebras is a non-singular irreducible subvariety of $\mathrm{Gr}(2,\mathfrak t^1_\Delta)$ isomorphic to $\mathbb{P}^2$. Indeed, in that case any 2-dimensional subspace of $\mathfrak t^1$ containing the central element
$c=r_1+r_2+t_{12}+s_{12}$ is a Gaudin subalgebra. The closures of the families (\ref{B}) and (\ref{A}) are $\mathbb{P}^1$-subvarieties of it. One can check that the first one is a conic, the second one is a line, intersecting in two points, which are $\mathrm{span}(c, r_1)$ and $\mathrm{span}(c, r_2).$

\emph{The $B_3$-Case: } The variety of Gaudin subalgebras is comprised of eight non-singular irreducible components of different dimensions. Among them are the two-dimensional closures of the families (\ref{B}) and (\ref{A}), which we denote by $B_{\{123\}}$ and $A_{\{123\}}$ respectively.
They are (a special case of) weak del Pezzo surface of degree 2 and the degree 5 del Pezzo surface respectively.

Three additional two-dimensional irreducible components correspond to the $B_2$-subsystems and each is isomorphic to $\mathbb{P}^2$. We call them $B_{\{12\}}, B_{\{13\}}$ and $B_{\{23\}}$. The Gaudin subalgebra corresponding to a point $(x_i : x_j : x_{ij})\in\mathbb{P}^2\cong B_{\{ij\}}$, for $1\leq i<j\leq 3$, is of the form
$$\mathrm{span}(c_\Delta, c_{ij}, x_i r_i+x_j r_j +x_{ij} t_{ij})\ ,$$
where $c_\Delta$ is the central element
$$
c_\Delta=\sum_{i=1}^3r_i+\sum_{i<j}^3(t_{ij}+s_{ij})
$$
 and $$c_{ij}=r_i+r_j+t_{ij}+s_{ij}.$$
 
 Finally there are three $\mathbb{P}^1$-components which we call $A_{\{12\}}, A_{\{13\}}$ and $A_{\{23\}}$. The Gaudin subalgebra corresponding to a point $(y_i : y_j)\in\mathbb{P}^1\cong A_{\{ij\}}$ is of the form
$$\mathrm{span}(c_\Delta, r_1+r_2+r_3 + t_{ij}+s_{ij}, y_i t_{ij}+y_j s_{ij}).$$

The irreducible components of the Gaudin variety intersect as follows: $B_{\{123\}}$ and $A_{\{123\}}$ each intersect $B_{\{ij\}}$ in a respective $\mathbb{P}^1$-curve. Those two curves intersect at two points 
$$\mathrm{span}(c_\Delta, c_{ij}, r_i),\,\mathrm{span}(c_\Delta, c_{ij}, r_j),$$
giving altogether six points, which are the only points of intersection of the components $B_{\{123\}}$ and $A_{\{123\}}$. Each of the curves $A_{\{ij\}}$ does not meet the component $B_{\{123\}}$ and intersects $A_{\{123\}}$ at a single point
$$\mathrm{span}(c_\Delta, r_1+r_2+r_3, t_{ij}+s_{ij}).$$

Notice that for the $B_3$-case, the variety of Gaudin subalgebras contains as components the Gaudin varieties of types $A_l$ and $B_l$ for any $l<3$. This result can be extended to the $B_r$-case where, apart from an irreducible $B_r$- and $A_r$-component, all Gaudin varieties of types $A_l$ and $B_l$, $l<r$ are contained (in multiple copies). The intersection structure of these subvarieties is quite intricate, but its combinatorics could in principle be deduced from the intersection lattice in a recursive manner. It is not completely clear whether one is able to capture all Gaudin subalgebras in this recursive way for $r>3$.

It is however remarkable that for $B_n$ there are no abelian subalgebras contained in $\mathfrak t^1_\Delta$ which have dimensions larger than $n$, i.e. Gaudin subalgebras are maximal exactly as in the $A_n$-case. This fact is owed to the arrangement being fibre-type. It is known that among the finite Coxeter arrangements only the $A_n$, $B_n$ and dihedral cases are of fibre-type and indeed in every other case there can be found abelian subalgebras of higher dimension than the rank of the arrangement (see \cite{Leoth}).




\section{Acknowledgements.}

We are grateful to I. Cherednik, M. Kapranov and T. Kohno for stimulating discussions.

The work of APV was partly supported by the EPSRC (grant EP/J00488X/1). The work of GF was
partly supported by the Swiss National Science Foundation (National Centre of Competence in Research ``The Mathematics of Physics -- SwissMAP'').

\end{document}